\documentclass[aps,pra,preprint,showpacs,groupedaddress]{revtex4-1}
\usepackage{amsthm}
\usepackage[english]{babel}
\usepackage[T1]{fontenc}
\usepackage[latin1]{inputenc}
\usepackage{multirow}

\def\beq{\begin{equation}}
\def\eeq{\end{equation}}
\def\beqn{\begin{eqnarray}}
\def\eeqn{\end{eqnarray}}

\begin{document}


\title{Number fluctuations of cold spatially split bosonic objects}



\author{Kaspar Sakmann$^{1}$\footnote{E-mail:
kaspar.sakmann@pci.uni-heidelberg.de},
Alexej I. Streltsov$^{1}$\footnote{E-mail: 
alexej.streltsov@pci.uni-heidelberg.de},
Ofir. E. Alon$^{2}$\footnote{E-mail:
ofir@research.haifa.ac.il},
and Lorenz S. Cederbaum$^{1}$\footnote{E-mail: 
lorenz.cederbaum@pci.uni-heidelberg.de}}

\affiliation{$^1$ Theoretische Chemie, Physikalisch-Chemisches Institut, 
Universit\"at Heidelberg,\\
Im Neuenheimer Feld 229, D-69120 Heidelberg, Germany}

\affiliation{$^2$ Department of Physics, University of Haifa at Oranim, 
Tivon 36006, Israel}



\date{\today}

\begin{abstract}
We investigate the number fluctuations of spatially split many-boson systems employing a 
theorem about the maximally and minimally attainable variances of an observable. 
The number fluctuations of many-boson systems are given 
for different numbers of lattice sites and both mean-field and many-body wave functions.
It is shown which states maximize the particle number fluctuations, both in lattices and double-wells.
The fragmentation of the states is discussed, and it is shown that the number fluctuations of 
some fragmented states are identical 
to those of fully condensed states.
\end{abstract}


\maketitle

\section{Introduction}
Ultracold atoms offer the unique possibility to directly 
compare theoretical predictions about many-body physics 
with experiments. Many-body effects arise due to the 
interaction between particles and the external
trapping potential. Of particular interest is the 
question about the nature of the quantum state present in a given system. 
For example, in double- and multi-well trapping geometries, the ground 
state is either fragmented or condensed, depending on 
the barrier height and the interaction strength 
\cite{SpekkensSipe,Cederbaum2003,OfirPathway,Alon2005,MCHB,Mueller,RDMpaper}.
For long-range interactions even the ground state 
in a single-well trap can be fragmented \cite{Fischer2}.
Apart from the fragmentation of the ground state,  
fragmentation of Bose-Einstein condensates (BECs) is also known to develop in nonequilibrium 
dynamics \cite{ZollerFragmentation,BarrierPRL,fragmenton,ExactBJJ,caton,
JulianNumberSqueezing,JulianInterferometry,Optimal,SolitonDeath,TwinAtomBeams}. 
On the theoretical level, fragmentation manifests 
itself in the reduced density matrices of the system.
While the reduced density matrices of a system of bosons
themselves are not experimentally accessible, 
it is possible to draw conclusions about them 
from the measurement of experimentally accessible quantities, such as 
the particle number fluctuations.

Here, we would like to investigate the number fluctuations
of several many-body states and their fragmentation. We focus on cold spatially split bosonic objects. 
Atom number fluctuations of fragmented and condensed systems 
have been investigated intensively  
both theoretically, see e.g. \cite{BouyerKasevich,SpekkensSipe,Wineland,Burnett,ZollerFragmentation,
JavanainenNumberFluctuations,JulianNumberSqueezing,JulianInterferometry} 
and experimentally \cite{Orzel,Greiner2002,Li,KetterleSqueezing,JeromeSqueezing,MarkusSubPoissonian,TwinAtomBeams}, 
to name just a few.  Here, we would like to elucidate the limits that 
quantum mechanics puts on number fluctuations of bosonic mean-field and many-body 
wave functions, and concentrate on their fragmentation.
In particular, we find that some fragmented
states cannot be distinguished by their number fluctuations from 
fully condensed superfluid states.

This paper is organized as follows. In Sec. \ref{secvarthm} 
we discuss a theorem about the maximum variance of an observable. 
In Sec. \ref{secdef} some basic definitions needed for the discussion of number fluctuations in 
multi-well traps are introduced. 
In Sec. \ref{sec2mode} we discuss the number fluctuations 
of two-mode many-boson states. In Sec. \ref{seclat} we 
show how the previously obtained results can be generalized to BECs 
in multi-well traps and optical lattices. We summarize the results and conclude in Sec. \ref{seccon}.

\section{Maximum Variance Theorem}\label{secvarthm}
In this section we will prove a theorem about 
the maxima and minima of the variance of an observable.  
A more mathematically oriented  proof can be found in Ref. \cite{Textor}.
\newtheorem{varthm}{Theorem}
\begin{varthm}\label{varthm}
Let $\hat A$ denote a hermitian operator and $a_0<\dots<a_N$ 
a subset of its spectrum.
Furthermore, let $\vert\Psi\rangle$ be a wave function and 
\beq
\vert\Psi\rangle=\sum_{n=0}^N C_n \vert \Psi_n \rangle 
\eeq
be an expansion of $\vert\Psi\rangle$ in $A$'s eigenstates 
with $\sum_{n=0}^N \vert C_n \vert^2=1$,
where $C_n \vert \Psi_n\rangle$ denotes the contribution of all degenerate eigenstates 
of $\hat A$ to the eigenvalue $a_n$.
Then the variance of $\hat A$
\beq\label{defvar}
\Delta A^2\equiv \langle \Psi\vert\left(\hat A- \langle \Psi\vert \hat A\vert \Psi\rangle\right)^2\vert\Psi\rangle
\eeq
takes on its minimum, $\Delta A^2_{min}=0$, for $\vert C_m\vert^2=1$,
where $m\in 0,1,\dots,N$ is arbitrary. The maximum variance, 
$\Delta A^2_{max}=\frac{1}{4}\left(a_N-a_0\right)^2$ is obtained 
for $\vert C_0\vert^2=\vert C_N\vert^2=1/2$.
\end{varthm}
\begin{proof}
The normalization constraint $\sum_{n=0}^N \vert C_n \vert^2=1$ can be used to eliminate 
one of the $N+1$ coefficients $C_n$ in Eq. (\ref{defvar}). 
Without loss of generality, we choose the coefficient $C_N$
and write $\vert C_{N}\vert^2 = 1-\sum_{n=0}^{N-1} \vert C_n\vert^2$. 
The variance $\Delta A^2$ can then be written as 
\beqn
\Delta A^2&=&  \langle \Psi\vert \hat A^2 \vert\Psi\rangle - \langle \Psi\vert \hat A \vert\Psi\rangle^2\nonumber\\ 
        &=&\sum_{n=0}^{N} \vert C_n\vert^2 a_n^2 - \left(\sum_{n=0}^{N}  \vert C_n\vert^2 a_n \right)^2 \nonumber\\
	&=& \sum_{n=0}^{N-1} \vert C_n\vert^2 (a_n-a_{N})^2-\left (\sum_{n=0}^{N-1} \vert C_n\vert^2 (a_n-a_{N})\right)^2\nonumber\\
	& &\equiv\tilde{\Delta}A^2,
\eeqn
where $\tilde \Delta A^2$ is a function of 
$2N$ independent parameters $C_0,\dots,C_{N-1},C_0^\ast,\dots,C_{N-1}^\ast$
and $\sum_{n=0}^{N-1} \vert C_n \vert^2 \le 1$.
For an extremum $\frac{\partial \tilde \Delta A^2}{\partial C_n^\ast}=0$ 
must hold which yields
\beqn\label{maxeq}
C_n\left(a_n-a_{N}\right)&=& 2 C_n\sum_{m=0}^{N-1}\vert C_m\vert^2 (a_m-a_{N})
\eeqn
for $n=0,\dots,N-1$. The conditions $\frac{\partial \tilde \Delta A^2}{\partial C_n}=0$ 
yield the complex conjugates of Eqs. (\ref{maxeq}). 
The set of Eqs. (\ref{maxeq}) have the solution $C_n=0$ 
for $n=0,\dots,N-1$. In this case $\vert C_N \vert=1$ and $\Delta A^2=0$, 
which is a minimum because $\Delta A^2\ge0$. Since the choice to eliminate 
the coefficient $C_N$ was arbitrary, any state with $\vert C_n\vert=1$ for $n\in 0,\dots,N$ 
minimizes $\Delta A^2$ with the value $\Delta A^2_{min}=0$. 
This concludes the proof of the first part of Theorem \ref{varthm}.

Now suppose that $C_i\neq 0$ for at least one $i\in 0,\dots,N-1$, then it follows from 
Eqs. (\ref{maxeq}) that
\beqn\label{maxieq}
\left(a_i-a_{N}\right)&=& 2 \sum_{m=0}^{N-1}\vert C_m\vert^2 (a_m-a_{N}).
\eeqn
In general Eq. (\ref{maxieq}) can only be fulfilled if $\vert C_i\vert^2=1/2$ and $C_m=0$ 
for all $m\in1,\dots,N-1$ with $m\neq i$. It then follows from the normalization condition that $\vert C_N\vert^2=1/2$.
Since the choices $i$ and $N$ were arbitrary, the maximum
variance is obtained for that pair of coefficients $(i,j)$ with $\vert C_i\vert^2=\vert C_j\vert^2=1/2$ 
that maximizes $\tilde \Delta A^2$. The maximum variance of $\hat A$ can therefore be written as 
\beq\label{deltamax}
\Delta A^2_{max}=\frac{1}{4}\left(a_j - a_i\right)^2
\eeq
Since the eigenvalues $a_n$ are ordered increasingly, 
the maximum value of $\Delta A_{max}^2$ 
is obtained for the choice $i=0$ and $j=N$, i.e. 
for $\vert C_0\vert^2=\vert C_N\vert^2=1/2$. This concludes the proof of Theorem \ref{varthm}. 
\end{proof}

\section{Definitions}\label{secdef}
In this section we briefly recall some definitions that will be useful for what follows. 
Let the operators $\hat b_i^\dagger$  and $\hat b_i$ denote the operators that create and annihilate 
a boson in the orbital $\phi_i$ and fulfill the usual bosonic
commutation relations $[\hat b_i,\hat b_j^\dagger]=\delta_{ij}$. 
The number operator of the orbital $\phi_i$ is $\hat n_i=\hat b_i^\dagger \hat b_i$.
For a permanent in which $N$ bosons reside in $s$ orbitals with $n_i$ bosons in the orbital $\phi_i$ 
we use the shorthand notation 
\beq
\vert 1^{n_1},2^{n_2},\dots,s^{n_s}\rangle = \frac{\hat b_1^{\dagger n_1}\hat b_2^{\dagger n_2} \cdots\hat b_s^{\dagger n_s}}{\sqrt{n_1!n_2!\cdots n_s!}}\vert 0 \rangle.
\eeq
The most general wave function of $N$ identical bosons expanded in $s$  orbitals then reads
\beqn\label{psis}
\vert\Psi\rangle&=&\sum_{n_1,\dots,n_s=0}^N C_{n_1 \dots n_s} 
\vert 1^{n_1},2^{n_2},\dots,s^{n_s}\rangle,
\eeqn  
with $n_1+\dots+n_s=N$.
Furthermore, we write 
\beq
\hat \Psi(x)=\sum_i \hat b_i \phi_i(x)
\eeq
for the bosonic field operator and $\vert \Psi \rangle$ for an
$N$-boson wave function. Then $\hat \rho (x)=\hat \Psi^\dagger(x)\hat \Psi(x)$ is the 
operator of the single-particle density $\rho(x)=\langle\hat \rho (x)\rangle$. 
The first-order reduced density matrix (RDM) is defined as 
\beqn
\rho^{(1)}(x\vert x') &\equiv& \langle \Psi \vert \hat \Psi^\dagger(x')\hat \Psi(x) \vert \Psi \rangle\nonumber\\
&=& \sum_i n^{(1)}_i \alpha_i^{(1)}(x)\alpha_i^{(1)\ast}(x')
\eeqn
and has eigenfunctions $\alpha_i^{(1)}(x)$ and 
eigenvalues $n^{(1)}_i$ which
are known as natural orbitals and natural occupation numbers, respectively. Explicitly,
\beq
\int dx' \rho^{(1)}(x\vert x') \alpha_i^{(1)}(x') = n^{(1)}_i \alpha_i^{(1)}(x)
\eeq
holds, where $n^{(1)}_1\ge n^{(1)}_2\ge\dots$ is assumed and $\sum_{i} n^{(1)}_i = N$.
If an eigenvalue $n^{(1)}_i={\mathcal O}(N)$ 
exists the system is said to be condensed \cite{Penrose}. 
If there is more than one such eigenvalue, the BEC 
is said to be fragmented \cite{Noz82,*NozBook}, 
see also \cite{MCHB,Mueller,RDMpaper}. 
The density fluctuations are given by 
\beq
\Delta \rho^2(x)=\langle  \hat \rho(x)^2 \rangle-\langle\hat \rho(x)\rangle^2.
\eeq
In practice, $\Delta \rho^2(x)$ must be integrated over 
some finite region of space. If  each orbital $\phi_i(x)$ 
is localized around $x=x_i$ and has little overlap 
with other orbitals, the integral of $\Delta \rho^2(x)$ over a region of space where 
$\phi_i(x)$ is not negligible can be approximated by 
\beq
\Delta n_i^2=\langle \hat n_i^2\rangle - \langle \hat n_i\rangle^2.
\eeq
The quantities $\Delta n_i^2$ are known as number fluctuations and will be discussed in the following.

\section{Two-mode states} \label{sec2mode}
Let us now investigate the number fluctuations and the fragmentation of 
some particular many-boson states constructed either from two localized modes, denoted
$\phi_L$ and $\phi_R$ with $\phi_L(x)=\phi_R(-x)$
or their  {\em gerade} and {\em ungerade} combinations, 
denoted $\phi_g(x)=\frac{1}{\sqrt{2}}[\phi_R(x)+\phi_L(x)]$
and $\phi_u(x)=\frac{1}{\sqrt{2}}[\phi_R(x)-\phi_L(x)]$. 
%
%
%
%
%
%
The particle number operator of the orbital $\phi_R$ can then be written as
\beqn
\hat n_R&=&\hat b_R^\dagger \hat b_R
=\frac{1}{2} \left(
\hat n_g +\hat n_u + \hat b_g^\dagger \hat b_u+
\hat b_u^\dagger \hat b_g
\right).
\eeqn
Furthermore, if only two modes are available 
it follows from $\hat n_L = N - \hat n_R$
that 
\beq
\Delta n_L^2 = \Delta n_R^2,
\eeq
irrespective of the quantum state.
We will therefore drop the site index in this section.
\subsection{Many-body states}
In recent theoretical work based on the time-dependent many-body Schrödinger equation
it was shown that superpositions of macroscopic quantum states can be created 
by scattering an attractively interacting BEC from a barrier \cite{CastinCatState,caton}. 
The resulting state is known as a {\em caton}
and has two dominant contributions in the basis of left and right localized orbitals.
We idealize this caton state here by
\beq
\vert \Psi_{cat}\rangle 
= \frac{1}{\sqrt{2}}\left(\vert L^N\rangle + \vert R^N\rangle\right),
\eeq
which is also known as a NOON  state, since it can be written as 
$ \frac{1}{\sqrt{2}}\left(\vert N,0\rangle + \vert 0,N\rangle\right)$ 
using the conventional number state notation.
For the state 
$\vert \Psi_{cat}\rangle$ we find that  the number fluctuations are given by 
\begin{equation}\label{varcaton}
\Delta n^2_{cat} = N^2/4.
\end{equation}
Since $\vert 0,N\rangle$ and $\vert N,0\rangle$ are the eigenstates of $\hat n_R$ and  $\hat n_L$,
corresponding to their minimal and maximal eigenvalues, 
it follows from Theorem \ref{varthm} that the state $\vert \Psi_{cat}\rangle$ is a state 
that maximizes the variance of $\hat n_R$ and $\hat n_L$. Note that
any state of the form $\frac{1}{\sqrt{2}}[\vert L^N\rangle + \exp(i\theta)\vert R^N\rangle]$ leads to 
the same number fluctuations $\Delta n^2=N^2/4$.
The first-order RDM of such states has two macroscopic 
eigenvalues $n_1^{(1)}=n_2^{(1)}=N/2$, and thus 
the  caton is a fragmented BEC that maximizes the number fluctuations.
Thus, a measurement of the number fluctuations is insensitive
to the relative phase $\theta$ between $\vert N,0\rangle$ and $\vert 0,N\rangle$.

This result should be compared to that of a  caton state in the basis of 
the {\em gerade} and {\em ungerade} orbitals $\phi_g(x)$ and $\phi_u(x)$
\beq
\vert \Psi_{g/u\,cat}\rangle 
=  \frac{1}{\sqrt{2}}\left(\vert g^N\rangle  + \vert u^N\rangle\right),
\eeq
which has exactly the same set of eigenvalues of the first-order RDM 
$n_1^{(1)}=n_2^{(1)}=N/2$, but much smaller number fluctuations which are given by 
\begin{equation}\label{vargucaton}
\Delta n^2_{g/u\,cat} =  \frac{1}{4}N.
\end{equation}
Note that any state of the form $\frac{1}{\sqrt{2}}[\vert g^N\rangle + \exp(i\theta)\vert u^N\rangle]$
has the same fragmentation and number fluctuations as $\vert \Psi_{g/u\,cat}\rangle$.
Moreover, it is easy to see that also any state of the form 
$\cos(\theta)\vert g^N\rangle + \sin(\theta)\vert u^N\rangle$ 
has the same number fluctuations, $\Delta n^2=N/4$.
So far, no  caton states have been reported in experiments.
Equations (\ref{varcaton}) and (\ref{vargucaton}) and the considerations above
clearly show that  caton states cannot be characterized uniquely by their number 
fluctuations, or their fragmentation ratios alone.

\subsection{Mean-field states}
Spatially split mean-field states that have received a lot of 
attention are the soliton train states
\beq
\vert \Psi_{st}^{+}\rangle =  \vert g^N\rangle,\quad\vert \Psi_{st}^{-}\rangle =  \vert u^N\rangle
\eeq
that describe spatially split, fully condensed BECs, i.e. condensates with $n_1^{(1)}=N$.
Soliton trains appear in the context of attractively interacting BECs 
within the framework of Gross-Pitaevskii theory. 
Interestingly, one finds for their number fluctuations 
\begin{equation}\label{varst}
\Delta n^2_{st}= \frac{1}{4}N
\end{equation}
which is exactly the same result as for the $\vert \Psi_{g/u\,cat}\rangle$ state. 
Similar to the case of caton states discussed above,
also the state $\vert u^N \rangle$ leads to $\Delta n^2= \frac{1}{4}N$.
Thus, a measurement of the number fluctuations alone does not allow to distinguish 
between the states $\vert \Psi_{g/u\,cat}\rangle$,  $\vert \Psi_{st}^+\rangle$ and $\vert \Psi_{st}^-\rangle$. 
A simultaneous measurement of number fluctuations and fragmentation would be 
necessary to narrow down the
number of possible states that the system was in.

Let us now turn to more general mean-field states. 
To this end we define parameterized two-mode operators
\beqn
\hat a_1(\theta)&=&\cos(\theta)\hat b_L + \sin(\theta)\hat b_R,\nonumber\\
\hat a_2(\theta)&=&-\sin(\theta)\hat b_L + \cos(\theta)\hat b_R
\eeqn
which can annihilate bosons either in localized or delocalized orbitals depending on the value of $\theta$,
and compute the number fluctuations of the general mean-field state $\vert\Psi_{MF}\rangle$ that 
can be constructed from the operators $\hat a_1^\dagger$ and $\hat a_2^\dagger$: 
\beq
\vert \Psi_{MF}\rangle 
=  \vert a_1^{n}(\theta), a_2^{N-n}(\theta)\rangle.
\eeq
The number fluctuations of $\vert\Psi_{MF}\rangle$ are given by
\beqn\label{varMF}
\Delta n^2_{MF}&=& \left[\frac{N}{4}+\frac{n(N-n)}{2}\right]\sin^2(2\theta).
\eeqn
The maximum of the number fluctuations $\Delta n^2_{MF}$
when considered as a function of $\theta$ and $n$ is given by
\begin{equation}\label{varfrag}
\max{\Delta n^2_{MF}}  = \frac{N^2}{8} +\frac{N}{4}
\end{equation}
which is obtained for $n=N/2$ and $\theta=\pi/4$. 
For these values of $n$ and $\theta$ the wave function $\vert \Psi_{MF}\rangle$ becomes
\beq\label{psifrag}
\vert \Psi_{frag}\rangle 
=  \vert g^{N/2}, u^{N/2}\rangle.
\eeq
The state $\vert \Psi_{frag}\rangle$ is a so called  
{\em fragmenton} state \cite{fragmenton}. 
For attractively interacting BECs it was recently shown that soliton train states 
can quickly loose their coherence and become spatially split, 
fragmented objects, like the fragmenton state \cite{SolitonDeath}.
Just like the two caton states discussed above,  fragmentons
are fragmented BECs. The state $\vert \Psi_{frag}\rangle$ is two-fold fragmented with 
$n_1^{(1)}=n_2^{(1)}=N/2$. Interestingly, the number fluctuations of the fragmenton $\vert \Psi_{frag}\rangle$
has contributions $\propto N$ and $\propto N^2$, see Eq. (\ref{varfrag}). 
The minima of $\Delta n^2_{MF}$ are 
obtained for $\theta=0,\pi/2,\dots$ with $n\in 0,\dots,N$ arbitrary. The corresponding states are
known as number states or Fock states
\beq\label{num}
\vert\Psi_{num}\rangle
=\vert L^{n},R^{N-n}  \rangle
\eeq
and have zero number fluctuations $\Delta n^2=0$. 
The fragmentation of number states depends on the number of particles 
in each localized orbital, and is given by $n_1^{(1)}=n$, $n_2^{(2)}=N-n$.
This concludes our discussion of two-mode systems.

\section{Lattice states}\label{seclat}
\subsection{General lattice states}\label{subsecgen}
We will now generalize the discussion to lattices with $s$ sites, 
denoted $i=1,2,\dots,s$ and corresponding localized orbitals $\phi_i(x)$. 
%
%
First, we will prove that the maximum 
number fluctuations in an $s$-site lattice are identical to those 
of a two-mode system, namely $\max \Delta n_i^2=N^2/4$.
For simplicity, we begin with a lattice of $s=3$ sites.  
The ansatz wave function, Eq. (\ref{psis}), then reads
\beqn\label{psi3}
\vert\Psi_3\rangle&=&\sum_{n_1,n_2,n_3=0}^N C_{n_1 n_2 n_3} 
\vert{ 1^{n_1}, 2^{n_2}, 3^{n_3}}
 \rangle,
\eeqn  
with $n_1+n_2+n_3=N$ and $\sum_{n_1,n_2,n_3} \vert C_{n_1 n_2 n_3}\vert^2=1$. 
We define
\beq\label{constr3}
\vert \overline C_{n_1}\vert^2\equiv\sum_{n_2=0}^{N-n_1} \vert C_{n_1,n_2,N-n_1-n_2}\vert^2
\eeq
and note that 
$\sum_{n_1=0}^N \vert\overline C_{n_1}\vert^2=1$.
The variance  $\Delta n^2_1$ of $\langle \hat n_1 \rangle$ 
can then be written as
\beqn
\Delta n^2_1	&=&\langle \left(\hat n_1 -\langle \hat n_1\rangle \right)^2\rangle\nonumber\\
		&=&\sum_{n_1=0}^N \vert \overline C_{n_1}\vert^2 n_1^2 
	           - \left(\sum_{n_1=0}^N \vert \overline C_{n_1}\vert^2 n_1\right)^2.
\eeqn
Analogous to our proof of Theorem \ref{varthm}, it follows that for a three-site lattice 
the maximum number fluctuations are $\Delta n^2_1=N^2/4$, obtained for 
$\vert\overline C_{n_1=0}\vert^2=\vert\overline C_{n_1=N}\vert^2=1/2$.
Similarly, the minimum number fluctuations are $\Delta n^2_1=0$, 
obtained for $\vert\overline C_{n_1}\vert^2=1$ for any $n_1\in 1,\dots,N$. 
Since the choice of the lattice site $i=1$ was arbitrary, 
the minimum and maximum number fluctuations at any lattice site  $i=1,2,3$
are $\Delta n_i^2=0$ and $\Delta n_i^2=N^2/4$, respectively.
Thus, we recover the same values for the minimum and maximum 
number fluctuations
as in the case of two lattice sites, see Sec. \ref{sec2mode}. 
For a lattice with  $s$ sites the same reasoning applies 
if  $\vert \overline C_{n_1}\vert^2$  is redefined as 
the sum over all absolute value squares of coefficients 
with exactly $n_1$ bosons at lattice site $i=1$ [see Eq. (\ref{Cnredefined}) below].
Thus, we find 
\beq
\max{\Delta n_i^2} = \frac{N^2}{4},\qquad \min{\Delta n_i^2} = 0,
\eeq
for the absolute maximum and minimum number fluctuations for lattices with $s$ sites.

In the present calculation no assumption was made about the symmetry of the wave function.
Hence, states that maximize the number fluctuations $\Delta n_i^2$
will generally have different number fluctuations at different lattice sites.
This can easily be seen by noticing that, e.g., the state
$\frac{1}{\sqrt{2}}\left(\vert N,0,0\rangle + \vert 0,0,N\rangle\right)$
has number fluctuations $\Delta n_{1}^2=N^2/4$ at the first, but  
$\Delta n_{2}^2=0$ at the second lattice site. 
States that maximize number fluctuations and possess the symmetry of the lattice will be treated next.

\subsection{Symmetry restricted lattice states}\label{subsectransl}
We will now require that all lattice sites be equivalent with mean occupation $\langle\hat n_i\rangle=N/s$.
Since all sites are assumed to be equivalent we will drop the site index from now on.
As  before, we begin with $s=3$ lattice sites.
It is easy to see that the three-site lattice caton state 
\beq\label{cat-3}
\vert\Psi_{cat-3}\rangle=\frac{1}{\sqrt{3}}\left(\vert 1^N\rangle +\vert 2^N\rangle +\vert 3^N\rangle\right)
\eeq
has mean occupation $\langle\hat n\rangle=N/3$ and number fluctuations 
\beq\label{deltacat-3}
\Delta n^2_{cat-3}=\frac{2}{9}N^2
\eeq
for all three sites. Its number fluctuations are slightly less than the 
maximal possible value $\Delta n^2=N^2/4$, and we will now show that these 
are also the maximum number fluctuations under the constraint of equivalent sites.

As an ansatz for the wave function on the lattice we use Eq. (\ref{psis}).
Let us focus again on the number fluctuations on one, say the first,  of the $s$ equivalent lattice sites
and define the quantities
\beq\label{Cnredefined}
\vert \overline C_{n_1}\vert^2\equiv\sum_{n_2=0}^{N-n_1}\cdots\sum_{n_{s-1}=0}^{N-n_1-\dots-n_{s-1}} 
\vert C_{n_1,\dots,N-n_1-\dots-n_{s-1}}\vert^2.
\eeq
The requirement of mean occupation $N/s$ on all lattice sites can be written as
\beq
\sum_{n_i=0}^N \vert \overline{C}_{n_i}\vert^2 n_i - \frac{N}{s} = 0
\eeq
for $i=1,\dots,s$. Using the equivalence of all sites, 
we can focus on the first lattice site, and after dropping the site index the problem reduces 
to finding the extremum of the functional 
\beqn\label{tau}
\tau[\{\overline{C}_n^\ast\},\{\overline{C}_n\}]&=& \sum_{n=0}^{N}\vert \overline{C}_{n}\vert^2 n^2 -
\left( \sum_{n=0}^N \vert \overline{C}_{n}\vert^2 n \right)^2 \nonumber\\
&& - \mu\left( \sum_{n=0}^N \vert \overline{C}_{n}\vert^2 n - \frac{N}{s}\right) 
\eeqn 
where the normalization $\sum_{n=0}^{N}\vert \overline{C}_{n}\vert^2=1$ is used. This normalization 
constraint can be used to eliminate $\vert \overline{C}_N \vert^2$ in Eq. (\ref{tau}), giving
\beqn
\tau &=& \sum_{n=0}^{N-1}\vert \overline{C}_{n}\vert^2 (N-n)^2 - 
\left(\sum_{n=0}^{N-1} \vert \overline{C}_{n}\vert^2 (N-n)\right)^2 + \nonumber\\
&&\mu\left(\sum_{n=0}^{N-1} \vert \overline{C}_{n}\vert^2 (N-n) - \frac{s-1}{s}N \right).
\eeqn
For an extremum $\partial \tau /\partial C^\ast_n = 0$ must hold, i.e.
\beq\label{extr}
0=\left[(N-n)^2-2(N-n)\left(\frac{s-1}{s}N\right) + \mu (N-n)\right]\overline{C}_n
\eeq
for $n=0,\dots,N-1$. If $\overline{C}_n=0$ for $n=0,\dots,N-1$, Eqs. (\ref{extr}) are satisfied, but 
it follows from the normalization that $\vert\overline C_N\vert^2=\vert C_{N,0,\dots,0}\vert^2=1$, i.e.
the ansatz wave function, Eq. (\ref{psis}), reduces to $\vert N,0,\dots,0\rangle$. 
Not all sites are equivalent in the state $\vert N,0,\dots,0\rangle$ and therefore 
there is no solution with $\overline{C}_n=0$ for $n=0,\dots,N-1$.
Thus at least one $\vert\overline{C}_n\vert^2$ must be nonzero for $n=1,\dots,N-1$.
Assuming one particular nonzero $\overline{C}_n$, Eq. (\ref{extr}) puts the constraint
\beq
\mu=N\frac{s-2}{s}+n
\eeq
for each value of $n\in 1,\dots,N-1$ on $\mu$. Obviously, this constraint can only be satisfied for at most one
$n$. Thus, solutions to Eqs. (\ref{extr}) must be of the form $\overline{C}_n\neq0$ and $\overline{C}_N\neq0$, 
and the normalization constraint becomes $\vert \overline{C}_n\vert^2 + \vert \overline{C}_N\vert^2=1$. 
Likewise, the requirement of mean occupation $N/s$ reads
$\frac{N}{s}=\vert\overline{C_n}\vert^2 n+ \vert\overline{C_N}\vert^2 N$.
The two conditions can be combined to express $\vert\overline{C_N}\vert^2$ and $\vert\overline{C_n}\vert^2$ as
\beqn\label{eqCN}
\vert\overline{C_N}\vert^2 &=& \frac{1}{s}\frac{N-sn}{N-n},\nonumber\\
\vert\overline{C_n}\vert^2 &=& 1-\frac{1}{s}\frac{N-sn}{N-n}
\eeqn
which in turn can be used to express $\Delta n^2=\vert\overline{C_n}\vert^2 n^2+ \vert\overline{C_N}\vert^2 N^2
-(\vert\overline{C_n}\vert^2 n+ \vert\overline{C_N}\vert^2 N)^2$ after some algebra as
\beq\label{eqdel}
\Delta n^2 = -n\frac{s(s-1)N}{s^2} + \frac{s-1}{s^2} N^2.
\eeq
%
The maximum  of  $\Delta n^2$ as a function of $n$ 
is obtained for $n=0$ with
\beq\label{deltamaxs}
\max{\Delta n^2}=\frac{s-1}{s^2}N^2.
\eeq
Substituting $n=0$ in Eqs. (\ref{eqCN}) we find that the maximum 
particle number fluctuations for states with equivalent lattice sites
are obtained for
\beq
\vert\overline{C_0}\vert^2=\frac{s-1}{s}\qquad\vert\overline{C_N}\vert^2=\frac{1}{s}.
\eeq

Let us now return to the three-site lattice caton state,  given in Eq. (\ref{cat-3}).
By setting $s=3$ in Eq. (\ref{deltamaxs}) and comparing 
the result to Eq. (\ref{deltacat-3}), we find that
$\vert \Psi_{cat-3}\rangle$ is a state that 
maximizes the particle number fluctuations under 
the constraint that all three sites are equivalent. 
Note that also states with nonzero relative phases  between the components of 
$\vert \Psi_{cat-3}\rangle$ would lead to the same number fluctuations, 
but the sites would generally not be equivalent then.
More generally, we find for the $s$-site caton state
\beq\label{cat-s}
\vert\Psi_{cat-s}\rangle=\frac{1}{\sqrt{s}}\left(\vert 1^N\rangle + \vert 2^N\rangle +\dots+\vert s^N\rangle\right)
\eeq
that the number fluctuations are given by
\beq
\Delta n_{cat-s}^2=\frac{s-1}{s^2}N^2.
\eeq
This means that $s$-site caton states maximize the number fluctuations under the constraint that all sites 
are equivalent, see Eq. (\ref{deltamaxs}). 

Let us now discuss $s$-site mean-field states. In case that $N/s$ is integer, it is easy to see
that the Mott insulating state
\beq
\vert\Psi_{MI}\rangle =\vert 1^{N/s},2^{N/s},\dots,s^{N/s}\rangle
\eeq
is a lattice state with equivalent sites that minimizes the number fluctuations with $\Delta n^2_{MI}=0$. Therefore, 
the complete range of number fluctuations under the constraint of equivalent sites is 
\beq\label{poss}
0\le \Delta n^2 \le \frac{s-1}{s^2}N^2.
\eeq
Lastly, we define $\hat b_g^\dagger=\frac{1}{\sqrt{s}}(b_1^\dagger+\dots+\hat b_s^\dagger)$
and discuss the superfluid lattice state
\beq
\vert \Psi_{sf}\rangle=\vert g^N\rangle.
\eeq
The ground state of noninteracting bosons in a lattice potential is of this form  
and its number fluctuations are given by 
\beq
\Delta n_{sf}^2=N\frac{s-1}{s^2}.
\eeq
By comparison with Eq. (\ref{poss}) it becomes clear that the superfluid 
state $\vert \Psi_{sf}\rangle$ is about in the middle 
of the range of possible number fluctuations. It is fully condensed 
and hence its first order RDM has only one macroscopic eigenvalue, 
$n_1^{(1)}=N$, i.e. there is no fragmentation.
The state $\vert \Psi_{sf}\rangle$ is by far the most intensively studied state 
and concludes our investigation here.

\section{Conclusions}\label{seccon}
We have studied the number fluctuations and the fragmentation 
of various many-boson states, focusing on ultracold spatially split systems. 
Number fluctuations are a key quantity in determining the state of a quantum system.
We have shown that there is a great indeterminacy if number fluctuations are considered alone.
Additional observables will have to be considered to allow for  conclusive results, e.g. 
the fragmentation. For an  overview of the obtained results 
please see Table \ref{table1}.
\begin{table}
	\begin{tabular}{| l | c | c | c | c | }
	\hline
	Object &  $\#$ sites &  $n^{(1)}_1$ & $\Delta n^2$ & $\max{\Delta n^2}$ \\
	\hline
	\hline
	Caton					& 2 	& $N/2$ 	& $N^2/4$ 	& \multirow{4}{*}{$N^2/4$}  \\
	g/u Caton                               & 2 	& $N/2$ 	&  $N/4$ 	&   \\
	Soliton trains                         	& 2 	& $N$ 		& $N/4$ 	&   \\
	Fragmenton				& 2	& $N/2$		& $N^2/8+N/4$ 	&   \\
	\hline
	Lattice caton                          	& s	& $N/s$		& $N^2 (s-1)/s^2$	& \multirow{3}{*}{$N^2 (s-1)/s^2$}  \\
	Mott-insulator                         	& s	& $N/s$		& $0$		&  \\
	Superfluid				& s	& $N$		& $N (s-1)/s^2$		&  \\
	\hline
	\end{tabular}
\caption{Number fluctuations and fragmentation of different spatially split bosonic objects. 
Given are the number of sites over which the object is distributed, 
the largest eigenvalue of the first-order reduced density matrix $n^{(1)}_1$, 
the number fluctuations $\Delta n^2$ and the maximally obtainable number fluctuations $\max{\Delta n^2}$. 
Only objects for which all sites are equivalent are shown. \label{table1}}
\end{table}

\begin{acknowledgments}
We thank M. K. Oberthaler for stimulating discussions that lead to this work.
Financial support by the DFG is gratefully acknowledged. 
\end{acknowledgments}

\newpage
\bibliography{references}

\begin{thebibliography}{33}%
\makeatletter
\providecommand \@ifxundefined [1]{%
 \@ifx{#1\undefined}
}%
\providecommand \@ifnum [1]{%
 \ifnum #1\expandafter \@firstoftwo
 \else \expandafter \@secondoftwo
 \fi
}%
\providecommand \@ifx [1]{%
 \ifx #1\expandafter \@firstoftwo
 \else \expandafter \@secondoftwo
 \fi
}%
\providecommand \natexlab [1]{#1}%
\providecommand \enquote  [1]{``#1''}%
\providecommand \bibnamefont  [1]{#1}%
\providecommand \bibfnamefont [1]{#1}%
\providecommand \citenamefont [1]{#1}%
\providecommand \href@noop [0]{\@secondoftwo}%
\providecommand \href [0]{\begingroup \@sanitize@url \@href}%
\providecommand \@href[1]{\@@startlink{#1}\@@href}%
\providecommand \@@href[1]{\endgroup#1\@@endlink}%
\providecommand \@sanitize@url [0]{\catcode `\\12\catcode `\$12\catcode
  `\&12\catcode `\#12\catcode `\^12\catcode `\_12\catcode `\%12\relax}%
\providecommand \@@startlink[1]{}%
\providecommand \@@endlink[0]{}%
\providecommand \url  [0]{\begingroup\@sanitize@url \@url }%
\providecommand \@url [1]{\endgroup\@href {#1}{\urlprefix }}%
\providecommand \urlprefix  [0]{URL }%
\providecommand \Eprint [0]{\href }%
\providecommand \doibase [0]{http://dx.doi.org/}%
\providecommand \selectlanguage [0]{\@gobble}%
\providecommand \bibinfo  [0]{\@secondoftwo}%
\providecommand \bibfield  [0]{\@secondoftwo}%
\providecommand \translation [1]{[#1]}%
\providecommand \BibitemOpen [0]{}%
\providecommand \bibitemStop [0]{}%
\providecommand \bibitemNoStop [0]{.\EOS\space}%
\providecommand \EOS [0]{\spacefactor3000\relax}%
\providecommand \BibitemShut  [1]{\csname bibitem#1\endcsname}%
\let\auto@bib@innerbib\@empty
\bibitem [{\citenamefont {Spekkens}\ and\ \citenamefont
  {Sipe}(1999)}]{SpekkensSipe}%
  \BibitemOpen
  \bibfield  {author} {\bibinfo {author} {\bibfnamefont {R.~W.}\ \bibnamefont
  {Spekkens}}\ and\ \bibinfo {author} {\bibfnamefont {J.~E.}\ \bibnamefont
  {Sipe}},\ }\href {\doibase 10.1103/PhysRevA.59.3868} {\bibfield  {journal}
  {\bibinfo  {journal} {Phys. Rev. A}\ }\textbf {\bibinfo {volume} {59}},\
  \bibinfo {pages} {3868} (\bibinfo {year} {1999})}\BibitemShut {NoStop}%
\bibitem [{\citenamefont {Cederbaum}\ and\ \citenamefont
  {Streltsov}(2003)}]{Cederbaum2003}%
  \BibitemOpen
  \bibfield  {author} {\bibinfo {author} {\bibfnamefont {L.~S.}\ \bibnamefont
  {Cederbaum}}\ and\ \bibinfo {author} {\bibfnamefont {A.~I.}\ \bibnamefont
  {Streltsov}},\ }\href {\doibase DOI: 10.1016/j.physleta.2003.09.058}
  {\bibfield  {journal} {\bibinfo  {journal} {Phys. Lett. A}\ }\textbf
  {\bibinfo {volume} {318}},\ \bibinfo {pages} {564 } (\bibinfo {year}
  {2003})}\BibitemShut {NoStop}%
\bibitem [{\citenamefont {Alon}\ and\ \citenamefont
  {Cederbaum}(2005)}]{OfirPathway}%
  \BibitemOpen
  \bibfield  {author} {\bibinfo {author} {\bibfnamefont {O.~E.}\ \bibnamefont
  {Alon}}\ and\ \bibinfo {author} {\bibfnamefont {L.~S.}\ \bibnamefont
  {Cederbaum}},\ }\href {\doibase 10.1103/PhysRevLett.95.140402} {\bibfield
  {journal} {\bibinfo  {journal} {Phys. Rev. Lett.}\ }\textbf {\bibinfo
  {volume} {95}},\ \bibinfo {pages} {140402} (\bibinfo {year}
  {2005})}\BibitemShut {NoStop}%
\bibitem [{\citenamefont {Alon}\ \emph {et~al.}(2005)\citenamefont {Alon},
  \citenamefont {Streltsov},\ and\ \citenamefont {Cederbaum}}]{Alon2005}%
  \BibitemOpen
  \bibfield  {author} {\bibinfo {author} {\bibfnamefont {O.~E.}\ \bibnamefont
  {Alon}}, \bibinfo {author} {\bibfnamefont {A.~I.}\ \bibnamefont {Streltsov}},
  \ and\ \bibinfo {author} {\bibfnamefont {L.~S.}\ \bibnamefont {Cederbaum}},\
  }\href {\doibase DOI: 10.1016/j.physleta.2005.06.118} {\bibfield  {journal}
  {\bibinfo  {journal} {Phys. Lett. A}\ }\textbf {\bibinfo {volume} {347}},\
  \bibinfo {pages} {88 } (\bibinfo {year} {2005})}\BibitemShut {NoStop}%
\bibitem [{\citenamefont {Streltsov}\ \emph {et~al.}(2006)\citenamefont
  {Streltsov}, \citenamefont {Alon},\ and\ \citenamefont {Cederbaum}}]{MCHB}%
  \BibitemOpen
  \bibfield  {author} {\bibinfo {author} {\bibfnamefont {A.~I.}\ \bibnamefont
  {Streltsov}}, \bibinfo {author} {\bibfnamefont {O.~E.}\ \bibnamefont {Alon}},
  \ and\ \bibinfo {author} {\bibfnamefont {L.~S.}\ \bibnamefont {Cederbaum}},\
  }\href {\doibase 10.1103/PhysRevA.73.063626} {\bibfield  {journal} {\bibinfo
  {journal} {Phys. Rev. A}\ }\textbf {\bibinfo {volume} {73}},\ \bibinfo
  {pages} {063626} (\bibinfo {year} {2006})}\BibitemShut {NoStop}%
\bibitem [{\citenamefont {Mueller}\ \emph {et~al.}(2006)\citenamefont
  {Mueller}, \citenamefont {Ho}, \citenamefont {Ueda},\ and\ \citenamefont
  {Baym}}]{Mueller}%
  \BibitemOpen
  \bibfield  {author} {\bibinfo {author} {\bibfnamefont {E.~J.}\ \bibnamefont
  {Mueller}}, \bibinfo {author} {\bibfnamefont {T.-L.}\ \bibnamefont {Ho}},
  \bibinfo {author} {\bibfnamefont {M.}~\bibnamefont {Ueda}}, \ and\ \bibinfo
  {author} {\bibfnamefont {G.}~\bibnamefont {Baym}},\ }\href {\doibase
  10.1103/PhysRevA.74.033612} {\bibfield  {journal} {\bibinfo  {journal} {Phys.
  Rev. A}\ }\textbf {\bibinfo {volume} {74}},\ \bibinfo {pages} {033612}
  (\bibinfo {year} {2006})}\BibitemShut {NoStop}%
\bibitem [{\citenamefont {Sakmann}\ \emph {et~al.}(2008)\citenamefont
  {Sakmann}, \citenamefont {Streltsov}, \citenamefont {Alon},\ and\
  \citenamefont {Cederbaum}}]{RDMpaper}%
  \BibitemOpen
  \bibfield  {author} {\bibinfo {author} {\bibfnamefont {K.}~\bibnamefont
  {Sakmann}}, \bibinfo {author} {\bibfnamefont {A.~I.}\ \bibnamefont
  {Streltsov}}, \bibinfo {author} {\bibfnamefont {O.~E.}\ \bibnamefont {Alon}},
  \ and\ \bibinfo {author} {\bibfnamefont {L.~S.}\ \bibnamefont {Cederbaum}},\
  }\href {\doibase 10.1103/PhysRevA.78.023615} {\bibfield  {journal} {\bibinfo
  {journal} {Phys. Rev. A}\ }\textbf {\bibinfo {volume} {78}},\ \bibinfo
  {pages} {023615} (\bibinfo {year} {2008})}\BibitemShut {NoStop}%
\bibitem [{\citenamefont {Bader}\ and\ \citenamefont
  {Fischer}(2009)}]{Fischer2}%
  \BibitemOpen
  \bibfield  {author} {\bibinfo {author} {\bibfnamefont {P.}~\bibnamefont
  {Bader}}\ and\ \bibinfo {author} {\bibfnamefont {U.~R.}\ \bibnamefont
  {Fischer}},\ }\href {\doibase 10.1103/PhysRevLett.103.060402} {\bibfield
  {journal} {\bibinfo  {journal} {Phys. Rev. Lett.}\ }\textbf {\bibinfo
  {volume} {103}},\ \bibinfo {pages} {060402} (\bibinfo {year}
  {2009})}\BibitemShut {NoStop}%
\bibitem [{\citenamefont {Menotti}\ \emph {et~al.}(2001)\citenamefont
  {Menotti}, \citenamefont {Anglin}, \citenamefont {Cirac},\ and\ \citenamefont
  {Zoller}}]{ZollerFragmentation}%
  \BibitemOpen
  \bibfield  {author} {\bibinfo {author} {\bibfnamefont {C.}~\bibnamefont
  {Menotti}}, \bibinfo {author} {\bibfnamefont {J.~R.}\ \bibnamefont {Anglin}},
  \bibinfo {author} {\bibfnamefont {J.~I.}\ \bibnamefont {Cirac}}, \ and\
  \bibinfo {author} {\bibfnamefont {P.}~\bibnamefont {Zoller}},\ }\href
  {\doibase 10.1103/PhysRevA.63.023601} {\bibfield  {journal} {\bibinfo
  {journal} {Phys. Rev. A}\ }\textbf {\bibinfo {volume} {63}},\ \bibinfo
  {pages} {023601} (\bibinfo {year} {2001})}\BibitemShut {NoStop}%
\bibitem [{\citenamefont {Streltsov}\ \emph {et~al.}(2007)\citenamefont
  {Streltsov}, \citenamefont {Alon},\ and\ \citenamefont
  {Cederbaum}}]{BarrierPRL}%
  \BibitemOpen
  \bibfield  {author} {\bibinfo {author} {\bibfnamefont {A.~I.}\ \bibnamefont
  {Streltsov}}, \bibinfo {author} {\bibfnamefont {O.~E.}\ \bibnamefont {Alon}},
  \ and\ \bibinfo {author} {\bibfnamefont {L.~S.}\ \bibnamefont {Cederbaum}},\
  }\href {\doibase 10.1103/PhysRevLett.99.030402} {\bibfield  {journal}
  {\bibinfo  {journal} {Phys. Rev. Lett.}\ }\textbf {\bibinfo {volume} {99}},\
  \bibinfo {pages} {030402} (\bibinfo {year} {2007})}\BibitemShut {NoStop}%
\bibitem [{\citenamefont {Streltsov}\ \emph {et~al.}(2008)\citenamefont
  {Streltsov}, \citenamefont {Alon},\ and\ \citenamefont
  {Cederbaum}}]{fragmenton}%
  \BibitemOpen
  \bibfield  {author} {\bibinfo {author} {\bibfnamefont {A.~I.}\ \bibnamefont
  {Streltsov}}, \bibinfo {author} {\bibfnamefont {O.~E.}\ \bibnamefont {Alon}},
  \ and\ \bibinfo {author} {\bibfnamefont {L.~S.}\ \bibnamefont {Cederbaum}},\
  }\href {\doibase 10.1103/PhysRevLett.100.130401} {\bibfield  {journal}
  {\bibinfo  {journal} {Phys. Rev. Lett.}\ }\textbf {\bibinfo {volume} {100}},\
  \bibinfo {pages} {130401} (\bibinfo {year} {2008})}\BibitemShut {NoStop}%
\bibitem [{\citenamefont {Sakmann}\ \emph {et~al.}(2009)\citenamefont
  {Sakmann}, \citenamefont {Streltsov}, \citenamefont {Alon},\ and\
  \citenamefont {Cederbaum}}]{ExactBJJ}%
  \BibitemOpen
  \bibfield  {author} {\bibinfo {author} {\bibfnamefont {K.}~\bibnamefont
  {Sakmann}}, \bibinfo {author} {\bibfnamefont {A.~I.}\ \bibnamefont
  {Streltsov}}, \bibinfo {author} {\bibfnamefont {O.~E.}\ \bibnamefont {Alon}},
  \ and\ \bibinfo {author} {\bibfnamefont {L.~S.}\ \bibnamefont {Cederbaum}},\
  }\href {\doibase 10.1103/PhysRevLett.103.220601} {\bibfield  {journal}
  {\bibinfo  {journal} {Phys. Rev. Lett.}\ }\textbf {\bibinfo {volume} {103}},\
  \bibinfo {pages} {220601} (\bibinfo {year} {2009})}\BibitemShut {NoStop}%
\bibitem [{\citenamefont {Streltsov}\ \emph {et~al.}(2009)\citenamefont
  {Streltsov}, \citenamefont {Alon},\ and\ \citenamefont {Cederbaum}}]{caton}%
  \BibitemOpen
  \bibfield  {author} {\bibinfo {author} {\bibfnamefont {A.~I.}\ \bibnamefont
  {Streltsov}}, \bibinfo {author} {\bibfnamefont {O.~E.}\ \bibnamefont {Alon}},
  \ and\ \bibinfo {author} {\bibfnamefont {L.~S.}\ \bibnamefont {Cederbaum}},\
  }\href {\doibase 10.1103/PhysRevA.80.043616} {\bibfield  {journal} {\bibinfo
  {journal} {Phys. Rev. A}\ }\textbf {\bibinfo {volume} {80}},\ \bibinfo
  {pages} {043616} (\bibinfo {year} {2009})}\BibitemShut {NoStop}%
\bibitem [{\citenamefont {Grond}\ \emph {et~al.}(2009)\citenamefont {Grond},
  \citenamefont {Schmiedmayer},\ and\ \citenamefont
  {Hohenester}}]{JulianNumberSqueezing}%
  \BibitemOpen
  \bibfield  {author} {\bibinfo {author} {\bibfnamefont {J.}~\bibnamefont
  {Grond}}, \bibinfo {author} {\bibfnamefont {J.}~\bibnamefont {Schmiedmayer}},
  \ and\ \bibinfo {author} {\bibfnamefont {U.}~\bibnamefont {Hohenester}},\
  }\href {\doibase 10.1103/PhysRevA.79.021603} {\bibfield  {journal} {\bibinfo
  {journal} {Phys. Rev. A}\ }\textbf {\bibinfo {volume} {79}},\ \bibinfo
  {pages} {021603} (\bibinfo {year} {2009})}\BibitemShut {NoStop}%
\bibitem [{\citenamefont {Grond}\ \emph {et~al.}(2010)\citenamefont {Grond},
  \citenamefont {Hohenester}, \citenamefont {Mazets},\ and\ \citenamefont
  {Schmiedmayer}}]{JulianInterferometry}%
  \BibitemOpen
  \bibfield  {author} {\bibinfo {author} {\bibfnamefont {J.}~\bibnamefont
  {Grond}}, \bibinfo {author} {\bibfnamefont {U.}~\bibnamefont {Hohenester}},
  \bibinfo {author} {\bibfnamefont {I.}~\bibnamefont {Mazets}}, \ and\ \bibinfo
  {author} {\bibfnamefont {J.}~\bibnamefont {Schmiedmayer}},\ }\href
  {http://stacks.iop.org/1367-2630/12/i=6/a=065036} {\bibfield  {journal}
  {\bibinfo  {journal} {New J. Phys.}\ }\textbf {\bibinfo {volume} {12}},\
  \bibinfo {pages} {065036} (\bibinfo {year} {2010})}\BibitemShut {NoStop}%
\bibitem [{\citenamefont {Sakmann}\ \emph {et~al.}(2011)\citenamefont
  {Sakmann}, \citenamefont {Streltsov}, \citenamefont {Alon},\ and\
  \citenamefont {Cederbaum}}]{Optimal}%
  \BibitemOpen
  \bibfield  {author} {\bibinfo {author} {\bibfnamefont {K.}~\bibnamefont
  {Sakmann}}, \bibinfo {author} {\bibfnamefont {A.~I.}\ \bibnamefont
  {Streltsov}}, \bibinfo {author} {\bibfnamefont {O.~E.}\ \bibnamefont {Alon}},
  \ and\ \bibinfo {author} {\bibfnamefont {L.~S.}\ \bibnamefont {Cederbaum}},\
  }\href {http://stacks.iop.org/1367-2630/13/i=4/a=043003} {\bibfield
  {journal} {\bibinfo  {journal} {New J. Phys.}\ }\textbf {\bibinfo {volume}
  {13}},\ \bibinfo {pages} {043003} (\bibinfo {year} {2011})}\BibitemShut
  {NoStop}%
\bibitem [{\citenamefont {Streltsov}\ \emph {et~al.}(2011)\citenamefont
  {Streltsov}, \citenamefont {Alon},\ and\ \citenamefont
  {Cederbaum}}]{SolitonDeath}%
  \BibitemOpen
  \bibfield  {author} {\bibinfo {author} {\bibfnamefont {A.~I.}\ \bibnamefont
  {Streltsov}}, \bibinfo {author} {\bibfnamefont {O.~E.}\ \bibnamefont {Alon}},
  \ and\ \bibinfo {author} {\bibfnamefont {L.~S.}\ \bibnamefont {Cederbaum}},\
  }\href {\doibase 10.1103/PhysRevLett.106.240401} {\bibfield  {journal}
  {\bibinfo  {journal} {Phys. Rev. Lett.}\ }\textbf {\bibinfo {volume} {106}},\
  \bibinfo {pages} {240401} (\bibinfo {year} {2011})}\BibitemShut {NoStop}%
\bibitem [{\citenamefont {Bucker}\ \emph {et~al.}(2011)\citenamefont {Bucker},
  \citenamefont {Grond}, \citenamefont {Manz}, \citenamefont {Berrada},
  \citenamefont {Betz}, \citenamefont {Koller}, \citenamefont {Hohenester},
  \citenamefont {Schumm}, \citenamefont {Perrin},\ and\ \citenamefont
  {Schmiedmayer}}]{TwinAtomBeams}%
  \BibitemOpen
  \bibfield  {author} {\bibinfo {author} {\bibfnamefont {R.}~\bibnamefont
  {Bucker}}, \bibinfo {author} {\bibfnamefont {J.}~\bibnamefont {Grond}},
  \bibinfo {author} {\bibfnamefont {S.}~\bibnamefont {Manz}}, \bibinfo {author}
  {\bibfnamefont {T.}~\bibnamefont {Berrada}}, \bibinfo {author} {\bibfnamefont
  {T.}~\bibnamefont {Betz}}, \bibinfo {author} {\bibfnamefont {C.}~\bibnamefont
  {Koller}}, \bibinfo {author} {\bibfnamefont {U.}~\bibnamefont {Hohenester}},
  \bibinfo {author} {\bibfnamefont {T.}~\bibnamefont {Schumm}}, \bibinfo
  {author} {\bibfnamefont {A.}~\bibnamefont {Perrin}}, \ and\ \bibinfo {author}
  {\bibfnamefont {J.}~\bibnamefont {Schmiedmayer}},\ }\href {\doibase
  10.1038/nphys1992} {\bibfield  {journal} {\bibinfo  {journal} {Nature Phys.}\
  }\textbf {\bibinfo {volume} {7}},\ \bibinfo {pages} {608} (\bibinfo {year}
  {2011})}\BibitemShut {NoStop}%
\bibitem [{\citenamefont {Bouyer}\ and\ \citenamefont
  {Kasevich}(1997)}]{BouyerKasevich}%
  \BibitemOpen
  \bibfield  {author} {\bibinfo {author} {\bibfnamefont {P.}~\bibnamefont
  {Bouyer}}\ and\ \bibinfo {author} {\bibfnamefont {M.~A.}\ \bibnamefont
  {Kasevich}},\ }\href {\doibase 10.1103/PhysRevA.56.R1083} {\bibfield
  {journal} {\bibinfo  {journal} {Phys. Rev. A}\ }\textbf {\bibinfo {volume}
  {56}},\ \bibinfo {pages} {R1083} (\bibinfo {year} {1997})}\BibitemShut
  {NoStop}%
\bibitem [{\citenamefont {Wineland}\ \emph {et~al.}(1992)\citenamefont
  {Wineland}, \citenamefont {Bollinger}, \citenamefont {Itano}, \citenamefont
  {Moore},\ and\ \citenamefont {Heinzen}}]{Wineland}%
  \BibitemOpen
  \bibfield  {author} {\bibinfo {author} {\bibfnamefont {D.~J.}\ \bibnamefont
  {Wineland}}, \bibinfo {author} {\bibfnamefont {J.~J.}\ \bibnamefont
  {Bollinger}}, \bibinfo {author} {\bibfnamefont {W.~M.}\ \bibnamefont
  {Itano}}, \bibinfo {author} {\bibfnamefont {F.~L.}\ \bibnamefont {Moore}}, \
  and\ \bibinfo {author} {\bibfnamefont {D.~J.}\ \bibnamefont {Heinzen}},\
  }\href {\doibase 10.1103/PhysRevA.46.R6797} {\bibfield  {journal} {\bibinfo
  {journal} {Phys. Rev. A}\ }\textbf {\bibinfo {volume} {46}},\ \bibinfo
  {pages} {R6797} (\bibinfo {year} {1992})}\BibitemShut {NoStop}%
\bibitem [{\citenamefont {Holland}\ and\ \citenamefont
  {Burnett}(1993)}]{Burnett}%
  \BibitemOpen
  \bibfield  {author} {\bibinfo {author} {\bibfnamefont {M.~J.}\ \bibnamefont
  {Holland}}\ and\ \bibinfo {author} {\bibfnamefont {K.}~\bibnamefont
  {Burnett}},\ }\href {\doibase 10.1103/PhysRevLett.71.1355} {\bibfield
  {journal} {\bibinfo  {journal} {Phys. Rev. Lett.}\ }\textbf {\bibinfo
  {volume} {71}},\ \bibinfo {pages} {1355} (\bibinfo {year}
  {1993})}\BibitemShut {NoStop}%
\bibitem [{\citenamefont {Javanainen}\ and\ \citenamefont
  {Ivanov}(1999)}]{JavanainenNumberFluctuations}%
  \BibitemOpen
  \bibfield  {author} {\bibinfo {author} {\bibfnamefont {J.}~\bibnamefont
  {Javanainen}}\ and\ \bibinfo {author} {\bibfnamefont {M.~Y.}\ \bibnamefont
  {Ivanov}},\ }\href {\doibase 10.1103/PhysRevA.60.2351} {\bibfield  {journal}
  {\bibinfo  {journal} {Phys. Rev. A}\ }\textbf {\bibinfo {volume} {60}},\
  \bibinfo {pages} {2351} (\bibinfo {year} {1999})}\BibitemShut {NoStop}%
\bibitem [{\citenamefont {Orzel}\ \emph {et~al.}(2001)\citenamefont {Orzel},
  \citenamefont {Tuchman}, \citenamefont {Fenselau}, \citenamefont {Yasuda},\
  and\ \citenamefont {Kasevich}}]{Orzel}%
  \BibitemOpen
  \bibfield  {author} {\bibinfo {author} {\bibfnamefont {C.}~\bibnamefont
  {Orzel}}, \bibinfo {author} {\bibfnamefont {A.~K.}\ \bibnamefont {Tuchman}},
  \bibinfo {author} {\bibfnamefont {M.~L.}\ \bibnamefont {Fenselau}}, \bibinfo
  {author} {\bibfnamefont {M.}~\bibnamefont {Yasuda}}, \ and\ \bibinfo {author}
  {\bibfnamefont {M.~A.}\ \bibnamefont {Kasevich}},\ }\href {\doibase
  10.1126/science.1058149} {\bibfield  {journal} {\bibinfo  {journal}
  {Science}\ }\textbf {\bibinfo {volume} {291}},\ \bibinfo {pages} {2386}
  (\bibinfo {year} {2001})}\BibitemShut {NoStop}%
\bibitem [{\citenamefont {Greiner}\ \emph {et~al.}(2002)\citenamefont
  {Greiner}, \citenamefont {Mandel}, \citenamefont {Esslinger}, \citenamefont
  {Hänsch},\ and\ \citenamefont {Bloch}}]{Greiner2002}%
  \BibitemOpen
  \bibfield  {author} {\bibinfo {author} {\bibfnamefont {M.}~\bibnamefont
  {Greiner}}, \bibinfo {author} {\bibfnamefont {O.}~\bibnamefont {Mandel}},
  \bibinfo {author} {\bibfnamefont {T.}~\bibnamefont {Esslinger}}, \bibinfo
  {author} {\bibfnamefont {T.~W.}\ \bibnamefont {Hänsch}}, \ and\ \bibinfo
  {author} {\bibfnamefont {I.}~\bibnamefont {Bloch}},\ }\href@noop {}
  {\bibfield  {journal} {\bibinfo  {journal} {Nature}\ }\textbf {\bibinfo
  {volume} {415}},\ \bibinfo {pages} {39} (\bibinfo {year} {2002})}\BibitemShut
  {NoStop}%
\bibitem [{\citenamefont {Li}\ \emph {et~al.}(2007)\citenamefont {Li},
  \citenamefont {Tuchman}, \citenamefont {Chien},\ and\ \citenamefont
  {Kasevich}}]{Li}%
  \BibitemOpen
  \bibfield  {author} {\bibinfo {author} {\bibfnamefont {W.}~\bibnamefont
  {Li}}, \bibinfo {author} {\bibfnamefont {A.~K.}\ \bibnamefont {Tuchman}},
  \bibinfo {author} {\bibfnamefont {H.-C.}\ \bibnamefont {Chien}}, \ and\
  \bibinfo {author} {\bibfnamefont {M.~A.}\ \bibnamefont {Kasevich}},\ }\href
  {\doibase 10.1103/PhysRevLett.98.040402} {\bibfield  {journal} {\bibinfo
  {journal} {Phys. Rev. Lett.}\ }\textbf {\bibinfo {volume} {98}},\ \bibinfo
  {pages} {040402} (\bibinfo {year} {2007})}\BibitemShut {NoStop}%
\bibitem [{\citenamefont {Jo}\ \emph {et~al.}(2007)\citenamefont {Jo},
  \citenamefont {Shin}, \citenamefont {Will}, \citenamefont {Pasquini},
  \citenamefont {Saba}, \citenamefont {Ketterle}, \citenamefont {Pritchard},
  \citenamefont {Vengalattore},\ and\ \citenamefont
  {Prentiss}}]{KetterleSqueezing}%
  \BibitemOpen
  \bibfield  {author} {\bibinfo {author} {\bibfnamefont {G.-B.}\ \bibnamefont
  {Jo}}, \bibinfo {author} {\bibfnamefont {Y.}~\bibnamefont {Shin}}, \bibinfo
  {author} {\bibfnamefont {S.}~\bibnamefont {Will}}, \bibinfo {author}
  {\bibfnamefont {T.~A.}\ \bibnamefont {Pasquini}}, \bibinfo {author}
  {\bibfnamefont {M.}~\bibnamefont {Saba}}, \bibinfo {author} {\bibfnamefont
  {W.}~\bibnamefont {Ketterle}}, \bibinfo {author} {\bibfnamefont {D.~E.}\
  \bibnamefont {Pritchard}}, \bibinfo {author} {\bibfnamefont {M.}~\bibnamefont
  {Vengalattore}}, \ and\ \bibinfo {author} {\bibfnamefont {M.}~\bibnamefont
  {Prentiss}},\ }\href {\doibase 10.1103/PhysRevLett.98.030407} {\bibfield
  {journal} {\bibinfo  {journal} {Phys. Rev. Lett.}\ }\textbf {\bibinfo
  {volume} {98}},\ \bibinfo {pages} {030407} (\bibinfo {year}
  {2007})}\BibitemShut {NoStop}%
\bibitem [{\citenamefont {Est\`eve}\ \emph {et~al.}(2008)\citenamefont
  {Est\`eve}, \citenamefont {Gross}, \citenamefont {Weller}, \citenamefont
  {Giovanazzi},\ and\ \citenamefont {Oberthaler}}]{JeromeSqueezing}%
  \BibitemOpen
  \bibfield  {author} {\bibinfo {author} {\bibfnamefont {J.}~\bibnamefont
  {Est\`eve}}, \bibinfo {author} {\bibfnamefont {C.}~\bibnamefont {Gross}},
  \bibinfo {author} {\bibfnamefont {A.}~\bibnamefont {Weller}}, \bibinfo
  {author} {\bibfnamefont {S.}~\bibnamefont {Giovanazzi}}, \ and\ \bibinfo
  {author} {\bibfnamefont {M.~K.}\ \bibnamefont {Oberthaler}},\ }\href
  {\doibase 10.1038/nature07332} {\bibfield  {journal} {\bibinfo  {journal}
  {Nature}\ }\textbf {\bibinfo {volume} {455}},\ \bibinfo {pages} {1216}
  (\bibinfo {year} {2008})}\BibitemShut {NoStop}%
\bibitem [{\citenamefont {Gross}\ \emph {et~al.}(2011)\citenamefont {Gross},
  \citenamefont {Est\`eve}, \citenamefont {Oberthaler}, \citenamefont
  {Martin},\ and\ \citenamefont {Ruostekoski}}]{MarkusSubPoissonian}%
  \BibitemOpen
  \bibfield  {author} {\bibinfo {author} {\bibfnamefont {C.}~\bibnamefont
  {Gross}}, \bibinfo {author} {\bibfnamefont {J.}~\bibnamefont {Est\`eve}},
  \bibinfo {author} {\bibfnamefont {M.~K.}\ \bibnamefont {Oberthaler}},
  \bibinfo {author} {\bibfnamefont {A.~D.}\ \bibnamefont {Martin}}, \ and\
  \bibinfo {author} {\bibfnamefont {J.}~\bibnamefont {Ruostekoski}},\ }\href
  {\doibase 10.1103/PhysRevA.84.011609} {\bibfield  {journal} {\bibinfo
  {journal} {Phys. Rev. A}\ }\textbf {\bibinfo {volume} {84}},\ \bibinfo
  {pages} {011609} (\bibinfo {year} {2011})}\BibitemShut {NoStop}%
\bibitem [{\citenamefont {Textor}(1978)}]{Textor}%
  \BibitemOpen
  \bibfield  {author} {\bibinfo {author} {\bibfnamefont {W.}~\bibnamefont
  {Textor}},\ }\href {http://dx.doi.org/10.1007/BF00673011} {\bibfield
  {journal} {\bibinfo  {journal} {Int. J. of Theor. Phys.}\ }\textbf {\bibinfo
  {volume} {17}},\ \bibinfo {pages} {599} (\bibinfo {year} {1978})}\BibitemShut
  {NoStop}%
\bibitem [{\citenamefont {Penrose}\ and\ \citenamefont
  {Onsager}(1956)}]{Penrose}%
  \BibitemOpen
  \bibfield  {author} {\bibinfo {author} {\bibfnamefont {O.}~\bibnamefont
  {Penrose}}\ and\ \bibinfo {author} {\bibfnamefont {L.}~\bibnamefont
  {Onsager}},\ }\href {\doibase 10.1103/PhysRev.104.576} {\bibfield  {journal}
  {\bibinfo  {journal} {Phys. Rev.}\ }\textbf {\bibinfo {volume} {104}},\
  \bibinfo {pages} {576} (\bibinfo {year} {1956})}\BibitemShut {NoStop}%
\bibitem [{\citenamefont {{Nozi\`eres, P.}}\ and\ \citenamefont {{Saint James,
  D.}}(1982)}]{Noz82}%
  \BibitemOpen
  \bibfield  {author} {\bibinfo {author} {\bibnamefont {{Nozi\`eres, P.}}}\
  and\ \bibinfo {author} {\bibnamefont {{Saint James, D.}}},\ }\href {\doibase
  10.1051/jphys:019820043070113300} {\bibfield  {journal} {\bibinfo  {journal}
  {J. Phys. France}\ }\textbf {\bibinfo {volume} {43}},\ \bibinfo {pages}
  {1133} (\bibinfo {year} {1982})}\BibitemShut {NoStop}%
\bibitem [{\citenamefont {Nozi\`eres}(1996)}]{NozBook}%
  \BibitemOpen
  \bibfield  {author} {\bibinfo {author} {\bibfnamefont {P.}~\bibnamefont
  {Nozi\`eres}},\ }\href@noop {} {\emph {\bibinfo {title} {Bose-Einstein
  Condensation}}},\ edited by\ \bibinfo {editor} {\bibfnamefont
  {A.}~\bibnamefont {Griffin}}, \bibinfo {editor} {\bibfnamefont {D.~W.}\
  \bibnamefont {Snoke}}, \ and\ \bibinfo {editor} {\bibfnamefont
  {S.}~\bibnamefont {Stringari}}\ (\bibinfo  {publisher} {Cambridge University
  Press, Cambridge, England},\ \bibinfo {year} {1996})\BibitemShut {NoStop}%
\bibitem [{\citenamefont {Weiss}\ and\ \citenamefont
  {Castin}(2009)}]{CastinCatState}%
  \BibitemOpen
  \bibfield  {author} {\bibinfo {author} {\bibfnamefont {C.}~\bibnamefont
  {Weiss}}\ and\ \bibinfo {author} {\bibfnamefont {Y.}~\bibnamefont {Castin}},\
  }\href {\doibase 10.1103/PhysRevLett.102.010403} {\bibfield  {journal}
  {\bibinfo  {journal} {Phys. Rev. Lett.}\ }\textbf {\bibinfo {volume} {102}},\
  \bibinfo {pages} {010403} (\bibinfo {year} {2009})}\BibitemShut {NoStop}%
\end{thebibliography}

\end{document}